\documentclass[11pt]{article}

\usepackage[margin=1in]{geometry}
\usepackage{amsmath,amsfonts,amssymb,amsthm}
\usepackage{enumitem}
\usepackage{url}
\usepackage[bookmarksnumbered,colorlinks,linkcolor=blue,citecolor=blue,urlcolor=blue]{hyperref}

\newcommand{\eprint}[1]{\href{http://arxiv.org/abs/#1}{#1}}

\newcommand{\C}{{\mathbb{C}}}
\newcommand{\F}{{\mathbb{F}}}
\newcommand{\R}{{\mathbb{R}}}

\newcommand{\tr}{\mathop{\mathrm{tr}}}
\newcommand{\spn}{\mathop{\mathrm{span}}}
\newcommand{\poly}{\mathop{\mathrm{poly}}}
\newcommand{\abs}{\mathop{\mathrm{abs}}}
\renewcommand{\d}{{\mathrm{d}}}
\newcommand{\ii}{{\mathrm{i}}}

\newcommand{\norm}[1]{\|{#1}\|}
\newcommand{\Norm}[1]{\left\|{#1}\right\|}
\newcommand{\floor}[1]{{\lfloor{#1}\rfloor}}
\newcommand{\ceil}[1]{{\lceil{#1}\rceil}}
\newcommand{\nint}[1]{{\lfloor{#1}\rceil}}
\newcommand{\amp}[2]{a_{{#1}|{#2}}}

\renewcommand{\>}{\rangle}
\newcommand{\<}{\langle}

\newcommand{\be}{\begin{equation}}
\newcommand{\ee}{\end{equation}}
\def\ba#1\ea{\begin{align}#1\end{align}}

\newtheorem{theorem}{Theorem}

\newcommand{\eq}[1]{(\ref{eq:#1})}
\renewcommand{\sec}[1]{Section~\ref{sec:#1}}
\newcommand{\thm}[1]{Theorem~\ref{thm:#1}}

\begin{document}
  

\title{On the relationship between continuous- \\ and discrete-time quantum walk}

\author{Andrew M.\ Childs\thanks{amchilds@uwaterloo.ca} \\[.5ex] Department of Combinatorics \& Optimization \\ and Institute for Quantum Computing \\ University of Waterloo}
\date{}

\maketitle

\begin{abstract}
Quantum walk is one of the main tools for quantum algorithms.  Defined by analogy to classical random walk, a quantum walk is a time-homogeneous quantum process on a graph.   Both random and quantum walks can be defined either in continuous or discrete time.  But whereas a continuous-time random walk can be obtained as the limit of a sequence of discrete-time random walks, the two types of quantum walk appear fundamentally different, owing to the need for extra degrees of freedom in the discrete-time case.

In this article, I describe a precise correspondence between continuous- and discrete-time quantum walks on arbitrary graphs.  Using this correspondence, I show that continuous-time quantum walk can be obtained as an appropriate limit of discrete-time quantum walks.  The correspondence also leads to a new technique for simulating Hamiltonian dynamics, giving efficient simulations even in cases where the Hamiltonian is not sparse.  The complexity of the simulation is linear in the total evolution time, an improvement over simulations based on high-order approximations of the Lie product formula.  As applications, I describe a continuous-time quantum walk algorithm for element distinctness and show how to optimally simulate continuous-time query algorithms of a certain form in the conventional quantum query model.  Finally, I discuss limitations of the method for simulating Hamiltonians with negative matrix elements, and present two problems that motivate attempting to circumvent these limitations.
\end{abstract}

\section{Introduction}
\label{sec:intro}

Recently, quantum walk has been established as one of the dominant algorithmic techniques for quantum computers.  In the black-box setting, quantum walk provides exponential speedup over classical computation \cite{CCDFGS03,CSV07}.  Moreover, many quantum walk algorithms achieve polynomial speedup over classical computation for problems of practical interest.  Following the development of quantum walk algorithms for search on graphs \cite{CG03,SKW02} (subsequently improved in \cite{AKR04,CG04,Tul08}), Ambainis cemented the importance of quantum walk by giving an optimal algorithm for the element distinctness problem \cite{Amb07}.  This approach was later generalized \cite{MNRS06,Sze04c} and applied to quantum algorithms for triangle finding \cite{MSS05}, checking matrix multiplication \cite{BS06}, and testing group commutativity \cite{MN07}.  More recently, Farhi, Goldstone, and Gutmann used quantum walk to develop an optimal algorithm for evaluating balanced binary game trees \cite{FGG07}, which led to optimal and near-optimal algorithms for evaluating broad classes of formulas \cite{ACRSZ07,RS07}.  Indeed, since Grover's well-known search algorithm \cite{Gro97} can be interpreted as a quantum walk on the complete graph, nearly all known quantum algorithms achieving polynomial speedup over classical computation can be viewed as quantum walk algorithms.

Quantum walk is defined by analogy to classical random walk, and is motivated by the ubiquity of random walk in classical randomized computation.
Random walks, or Markov chains, can either evolve continuously in time or by a discrete sequence of steps.  The relationship between these two notions of random walk is straightforward, and indeed it is often possible to analyze both settings simultaneously (see for example \cite{AF02}).

A discrete-time Markov chain with $N$ states is specified by an $N \times N$ stochastic matrix $M$, a matrix with nonnegative entries whose columns sum to $1$.  One step of the walk transforms an initial probability distribution $p \in \R^N$ into a new probability distribution $p'=Mp$.  To approach a continuous-time Markov process, consider a \emph{lazy random walk} in which we only make a transition with some small probability $\epsilon > 0$, replacing $M$ by $\epsilon M + (1-\epsilon) I$.  Then we have
\be
  p'-p = \epsilon (M-I) p. \label{eq:lazyrandom}
\ee
In the limit $\epsilon \to 0$, letting one step of the discrete-time walk correspond to a time interval $\epsilon$, we obtain the dynamics
\be
  \frac{\d}{\d{t}} p(t) = (M-I) p(t),
\label{eq:ctsclassical}
\ee
a continuous-time Markov process generated by $M-I$.  In fact, any continuous-time Markov chain with $N$ sites can be written in the form \eq{ctsclassical} for some $N \times N$ stochastic matrix $M$, up to a rescaling of the time variable.\footnote{For example, given a graph with adjacency matrix $A$, maximum degree $d$, and a diagonal matrix of vertex degrees $D$ (with $D_{jj}=\deg(j)$), the continuous-time random walk generated by the normalized Laplacian $(A-D)/d$ (a negative semidefinite operator) can be viewed as the limit of the discrete-time random walk with transition matrix $M = A/d + (I-D/d)$.  The simple discrete-time random walk, in which a transition is made to a randomly selected neighbor at each step, has the transition matrix $M = A D^{-1}$, and corresponds to the continuous-time random walk generated by $(A-D)D^{-1}$.  For a $d$-regular graph, $D=dI$, so these two choices are equivalent.}

Similarly, one can define two notions of quantum walk, the \emph{continuous-time quantum walk} \cite{CFG02,FG98} and the \emph{discrete-time quantum walk} \cite{AAKV01,ABNVW01,Wat01a}.  But in contrast to the classical case, these two models are apparently incomparable.

The continuous-time quantum walk on an undirected graph is obtained by replacing the diffusion equation \eq{ctsclassical} with the Schr{\"o}dinger equation
\be
  \ii \frac{\d}{\d{t}} q(t) = H q(t).
\label{eq:ctsquantum}
\ee
Here the Hamiltonian $H$ is an $N \times N$ Hermitian matrix with $H_{jk}\ne0$ if and only if vertices $j$ and $k$ are connected, such as the adjacency matrix or Laplacian of the graph, and $q(t) \in \C^N$ is a vector of complex amplitudes, one for each vertex.  The state space for the quantum walk is the same as for the corresponding classical random walk, and the dynamics are a direct quantum analog of \eq{ctsclassical}.

A step of a discrete-time quantum walk is a unitary operation that moves amplitude between adjacent vertices of a graph.  Unlike in continuous time, there is no general definition of a discrete-time quantum walk that takes place directly on the vertices.
In particular, there is no translation-invariant discrete-time unitary process on a $d$-dimensional lattice \cite{Mey96a,Mey96b}, and in fact, only graphs with special properties can support local unitary dynamics at all, even without homogeneity restrictions \cite{Sev03}.  However, by enlarging the state space to include extra degrees of freedom, sometimes referred to as a quantum coin---or equivalently, by considering a walk on the directed edges of a graph---this limitation can be overcome \cite{Wat01a}.

Despite these differences, examples of continuous- and discrete-time quantum walks on graphs show many qualitative similarities.  For example, the walks on the line both spread linearly in time \cite{ABNVW01,FG98}, and the walks on the hypercube have the same instantaneous mixing time \cite{MR01}.  Furthermore, some quantum algorithms have been constructed that behave similarly in the two models (compare \cite{CG03,CG04} to \cite{AKR04,SKW02}, and \cite{FGG07} to \cite{ACRSZ07}).  However, since the discrete-time walk takes place on a larger state space, it cannot simply reduce to the continuous-time walk as a limiting case.  Indeed, no general correspondence between the dynamics of the two models has been given previously (although a particular correspondence has been noted for the infinite line and for the three-dimensional square lattice \cite{Str06}).

From a computational perspective, the two types of quantum walk each have their own advantages.  As it does not require enlarging the state space, the continuous-time model is arguably more natural; it is simpler to define and usually easier to analyze.  For example, while the exponential speedup by continuous-time quantum walk described in \cite{CCDFGS03} is generally suspected to carry over to the discrete-time model, the dynamics are more involved, and to the best of my knowledge no rigorous analysis has been provided so far.

On the other hand, the discrete-time model has the advantage that it is generally easy to implement using quantum circuits, whereas the main technique for implementing a continuous-time quantum walk requires the maximum degree of the graph to be small \cite{AT03,BACS05,Chi04,CCDFGS03}.  Thus, for example, the element distinctness algorithm \cite{Amb07} seems difficult to reproduce using continuous-time quantum walk.

In this article, I explore the relationship between continuous- and discrete-time quantum walk, showing how to carry the desirable features of each model to the other.  The foundation of the results is a formal correspondence between the two models described in \sec{szegedy}.  This correspondence employs the quantization of discrete-time Markov chains proposed by Szegedy \cite{Sze04c}, which can also be used to give a discrete-time quantum walk corresponding to any Hermitian matrix, generalizing a construction of \cite{ACRSZ07,RS07}.

Using this correspondence, \sec{limit} introduces a discrete-time quantum walk whose behavior approaches that of a related continuous-time quantum walk in a certain limit.  In addition to providing a conceptual link between the two models, this limit offers a generic means of converting continuous-time quantum walk algorithms into discrete-time ones.  For example, it shows that there is an efficient discrete-time quantum walk algorithm for the graph traversal problem of \cite{CCDFGS03}.

More significantly, I apply the correspondence to the simulation of Hamiltonian dynamics by quantum circuits, giving more efficient simulations than have been known so far.  For a sparse Hamiltonian, it is well-known that the evolution according to \eq{ctsquantum} for time $t$ can be simulated in $t^{1+o(1)}$ steps \cite{BACS05,Chi04}.  By reduction to the problem of computing parity, reference \cite{BACS05} established what might be called a \emph{no fast-forwarding theorem}: in general, a Hamiltonian cannot be simulated for time $t$ using a number of steps that is sublinear in $t$.  However, this left open the question of whether a truly linear-time simulation is possible.  Moreover, even less has been known about the case where the Hamiltonian is not necessarily sparse, as introduced in \sec{nonsparse}.  I show in \sec{linear} that by applying phase estimation to a discrete-time quantum walk, one can simulate Hamiltonian dynamics for time $t$ using $O(t)$ operations, not only in the sparse case, but for even more general succinctly specified Hamiltonians.  As an example of this method in action, I give an optimal continuous-time quantum walk algorithm for element distinctness (\sec{ed}).  I also describe an application to the simulation of continuous-time query algorithms using conventional quantum queries (\sec{hamquery}).

Unfortunately, although the new simulation method can be efficient well beyond the case of sparse Hamiltonians, it can sometimes be inefficient if the Hamiltonian is not only dense, but also includes matrix elements of both signs (or, more generally, if it has complex entries).  I conclude in \sec{signproblem} by posing some examples of Hamiltonians the method fails to address, but whose efficient simulation would lead to new quantum algorithms.

\section{A discrete-time quantum walk for any Hamiltonian}
\label{sec:szegedy}

We begin by constructing a discrete-time quantum walk corresponding to an arbitrary Hermitian matrix, i.e., to the Hamiltonian of any finite-dimensional quantum system.  This construction is based on the work of Szegedy, who defined a discrete-time quantum walk corresponding to an arbitrary discrete-time classical Markov chain \cite{Sze04c}.  It was observed in \cite{ACRSZ07} that Szegedy's framework can be used to connect continuous- and discrete-time quantum walk, assuming that the continuous-time quantum walk is generated by an entrywise positive, symmetric matrix.  In \cite{RS07} this was generalized to an arbitrary Hermitian matrix whose graph of nonzero entries is bipartite.  Here, we describe a discrete-time quantum walk corresponding to a general $N \times N$ Hermitian matrix $H$.

Fix an orthonormal basis $\{|j\>: j=1,\ldots,N\}$ of $\C^N$, and let $\abs(H) := \sum_{j,k=1}^N |H_{jk}| \, |j\>\<k|$ denote the elementwise absolute value of $H$ in that basis.  Let $|d\> := \sum_{j=1}^N d_j |j\>$ be a principal eigenvector of $\abs(H)$ (i.e., an eigenvector with eigenvalue $\norm{\abs(H)}$).  If $\abs(H)$ is irreducible, then by the Perron-Frobenius theorem, $|d\>$ is unique up to a phase, and that phase can be chosen so that the entries $d_j$ are strictly positive.  We focus on the irreducible case without loss of generality, as if $\abs(H)$ is reducible, we can treat each of its irreducible components separately.  Note that $\norm{\abs(H)}\ge\norm{H}$, as can easily be proved using the triangle inequality.

Define a set of $N$ quantum states $|\psi_1\>,\ldots,|\psi_N\> \in \C^N \otimes \C^N$ as
\begin{align}
 |\psi_j\> &:= \frac{1}{\sqrt{\norm{\abs(H)}}} \sum_{k=1}^N \sqrt{H^*_{jk} \frac{d_k}{d_j}} \, |j,k\>. 
\label{eq:jstates}
\end{align}
It is straightforward to check that these states are orthonormal.
Notice that $\<j,k|\psi_j\> \ne 0$ if and only if $j$ is adjacent to $k$ in the graph of nonzero entries of $H$.

The discrete-time quantum walk corresponding to $H$ is defined as the unitary operator obtained by first reflecting about $\spn\{|\psi_j\>\}$ and then exchanging the two registers with the swap operation $S$ (i.e., $S|j,k\>=|k,j\>$).  Equivalently, two steps of the walk can be viewed as first reflecting about $\spn\{|\psi_j\>\}$, and then reflecting about $\spn\{S|\psi_j\>\}$.  Szegedy showed that the spectrum of a product of reflections depends in a simple way on the matrix of inner products between orthonormal bases for the two subspaces \cite[Theorem~1]{Sze04c}.\footnote{Note that this can also be viewed as a consequence of classic result of Jordan \cite{Reg06}.}  In the present case, we have
\ba
  \<\psi_j|S|\psi_k\>
  &= \frac{H_{jk}}{\norm{\abs(H)}}
\ea
(using the fact that $H=H^\dag$), so we obtain a walk corresponding to $H$.

More concretely, let
\be
  T:=\sum_{j=1}^N |\psi_j\>\<j|
\label{eq:isometry}
\ee
be the isometry mapping $|j\> \in \C^N$ to $|\psi_j\> \in \C^N \otimes \C^N$.  Then $TT^\dag$ is the projector onto $\spn\{|\psi_j\>\}$, so the walk operator described above is $S(2TT^\dag-1)$.  We have

\begin{theorem}\label{thm:spectrum}
Suppose that $\frac{H}{\norm{\abs(H)}}|\lambda\>=\lambda|\lambda\>$.
Then the unitary operator
\be
  U:=\ii S(2TT^\dag - 1)
\label{eq:discreteqwalk}
\ee
has two normalized eigenvectors
\be
  |\mu_\pm\> :=
  \frac{1 - e^{\pm\ii\arccos \lambda} S}
       {\sqrt{2(1 - \lambda^2)}}
  T|\lambda\>
\label{eq:discreteqwalkevec}
\ee
with eigenvalues $\mu_\pm := \pm e^{\pm\ii\arcsin\lambda}$.
\end{theorem}

This result follows directly from \cite[Theorem~1]{Sze04c}, but the proof can be simplified since the walk is defined using the swap operation.  We provide a proof here for the sake of completeness, along the same lines as \cite[Theorem~6]{ACRSZ07}.

\begin{proof}
Consider the action of $U$ on the vector $T|\lambda\>$.  Using $T^\dag T = 1$, we have
\ba
  U T|\lambda\> &= \ii S T|\lambda\>,
\ea
and using $T^\dag S T = H/\norm{\abs(H)}$, we have
\ba
  U S T|\lambda\> &=2\ii\lambda S T |\lambda\> - \ii T|\lambda\>.
\ea
Thus, the unnormalized state $|\mu\> := T|\lambda\> + \ii \mu ST|\lambda\>$ has
\be
  U|\mu\> = \mu T|\lambda\> + \ii(1+2\ii\lambda\mu) ST|\lambda\>,
\ee
and is an eigenvector of $U$ with eigenvalue $\mu$ provided $1+2\ii\lambda\mu = \mu^2$, i.e.,
\be
  \mu = \pm \sqrt{1-\lambda^2} + \ii \lambda
      = \ii e^{\mp \ii \arccos(\lambda)}
      = \pm e^{\pm \ii \arcsin(\lambda)}
\ee
as claimed.  The normalization follows from
\ba
  \<\mu|\mu\>
  = 1 + \ii\lambda(\mu - \mu^*) + |\mu|^2
  = 2(1-\lambda^2),
\ea
which completes the proof.
\end{proof}

\section{Continuous-time walk as a limit of discrete-time walks}
\label{sec:limit}

Using the correspondence described in the previous section, we can construct a discrete-time quantum walk whose behavior reproduces that of the continuous-time quantum walk generated by $H$ in an appropriate limit.  First we construct a discrete-time process that approximates the continuous-time one provided its eigenvalues are sufficiently small; then we construct a ``lazy quantum walk'' to obtain small eigenvalues.

Let $\Pi$ denote the projector onto the subspace $\spn\{T|j\>,ST|j\>: j=1,\ldots,N\}$, and consider the action of $U$ restricted to this (invariant) subspace.  If all eigenvalues $\lambda$ of $H/\norm{\abs(H)}$ are small, then $\arcsin\lambda \approx \lambda$, meaning that the eigenvalues of $H$ are approximately linearly related to the eigenphases of $U$.  Furthermore, the eigenvectors of $U$ are $|\mu_\pm\> \approx (1 \mp \ii S + \lambda S)T|\lambda\>/\sqrt2$, and we have
\be
  \ii \Pi U \Pi \approx \sum_{\lambda,\pm} \mp e^{\pm \ii \lambda} \frac{1 \mp \ii S + \lambda S}{\sqrt2} T|\lambda\>\<\lambda|T^\dag \frac{1 \pm \ii S + \lambda S}{\sqrt2}.
\ee
This expression appears similar to the spectral expansion of the evolution according to $H/\norm{\abs(H)}$ for unit time,
\be
  e^{-\ii H/\norm{\abs(H)}} = \sum_\lambda e^{-\ii \lambda} |\lambda\>\<\lambda|,
\ee
except for (i) application of the isometry $T$, (ii) the presence of both positive and negative phases for each eigenvalue $\lambda$, and (iii) rotations by approximately a square root of $\ii S$, namely $(1 \pm \ii S)/\sqrt 2$.  To obtain only the `$-$' terms, we rotate the basis by $(1+\ii S)/\sqrt 2$.  Then we find
\ba
   T^\dag \frac{1-\ii S}{\sqrt 2} (\ii U)^\tau \frac{1+\ii S}{\sqrt 2} T
   &\approx e^{-\ii H \tau},
\ea
where the calculation has been carried out to $O(\lambda)$.

More precisely, we have the following approximation:
\begin{theorem}\label{thm:dynamics}
Let $h := \norm{H}/\norm{\abs(H)}$.  Then
\be
  \Norm{T^\dag \frac{1-\ii S}{\sqrt 2} (\ii U)^\tau \frac{1+\ii S}{\sqrt 2} T - e^{-\ii \tau H/\norm{\abs(H)}}} \le h^2 \big( 1 + (\tfrac{\pi}{2}-1)h\tau \big).
  \label{eq:dynamicsestimate}
\ee
\end{theorem}

\begin{proof}
By \thm{spectrum}, we have
\ba
  \ii \Pi U \Pi
  &= \sum_{\lambda,\pm} \mp e^{\pm\ii\arcsin\lambda}
     \frac{1-e^{\pm\ii\arccos\lambda}S}{\sqrt{2(1-\lambda^2)}}
     T|\lambda\>\<\lambda|T^\dag
     \frac{1-e^{\mp\ii\arccos\lambda}S}{\sqrt{2(1-\lambda^2)}}.
\ea
Then a direct calculation gives
\ba
  &T^\dag \frac{1-\ii S}{\sqrt 2} (\ii U)^\tau \frac{1+\ii S}{\sqrt 2} T \nonumber\\
  &\quad\begin{aligned}= \frac{1}{4} \sum_{\lambda,\pm}
     \frac{\mp e^{\pm\ii \tau \arcsin\lambda}}{1-\lambda^2}
     \begin{aligned}[t]
     T^\dag&[(1+\ii e^{\pm\ii\arccos\lambda})-(\ii+e^{\pm\ii\arccos\lambda})S]T
     \,|\lambda\>\<\lambda| \\
     T^\dag&[(1-\ii e^{\mp\ii\arccos\lambda})+(\ii-e^{\mp\ii\arccos\lambda})S]T
     \end{aligned}\end{aligned} \\
  &\quad\begin{aligned}= \frac{1}{4} \sum_{\lambda,\pm}
     \frac{\mp e^{\pm\ii \tau \arcsin\lambda}}{1-\lambda^2}
     \begin{aligned}[t]
     &[(1+\ii e^{\pm\ii\arccos\lambda})-(\ii+e^{\pm\ii\arccos\lambda})\lambda] \\
     &[(1-\ii e^{\mp\ii\arccos\lambda})+(\ii-e^{\mp\ii\arccos\lambda})\lambda]
     |\lambda\>\<\lambda|
     \end{aligned}\end{aligned} \\
  &\quad\begin{aligned}[t]= \frac{1}{2} \sum_\lambda
     [-e^{+\ii \tau \arcsin\lambda} (1-\sqrt{1-\lambda^2})
      +e^{-\ii \tau \arcsin\lambda} (1+\sqrt{1-\lambda^2})]
     |\lambda\>\<\lambda|.
   \end{aligned}
\ea
Now we can use the inequality $1-\sqrt{1-\lambda^2} \le \lambda^2$ to bound the norm of the first term by $\lambda^2/2$.  By the same inequality, $|(1+\sqrt{1-\lambda^2})-2|/2 \le \lambda^2/2$.  Thus we have
\be
  \Norm{T^\dag \frac{1-\ii S}{\sqrt 2} (\ii U)^\tau \frac{1+\ii S}{\sqrt 2} T - e^{-\ii \tau \arcsin(H/\norm{\abs(H)})}} \le h^2.
\ee
Using the inequalities
$
  |\lambda-\arcsin\lambda| \le (\tfrac{\pi}{2}-1)|\lambda|^3
$
and $|1-e^{-\ii\theta}| \le |{\theta}|$, we find
\be
  \Norm{e^{-\ii \tau \arcsin(H/\norm{\abs(H)})}-e^{-\ii \tau H/\norm{\abs(H)}}}
  \le (\tfrac{\pi}{2}-1)\tau h^3,
\ee
and \eq{dynamicsestimate} follows by the triangle inequality.
\end{proof}

To obtain an arbitrarily good approximation of the dynamics, we require a systematic means of making $h = \norm{H}/\norm{\abs(H)}$ arbitrarily small.  In other words, we must construct a \emph{lazy quantum walk}, analogous to the lazy random walk that only takes a step with some small probability $\epsilon$ as in \eq{lazyrandom}.  To do this, we enlarge the Hilbert space and modify the states $|\psi_j\>$ from \eq{jstates} to
\be
  |\psi_j^\epsilon\> := \sqrt{\epsilon} |\psi_j\> + \sqrt{1-\epsilon} |{\perp_j}\>
\ee
for some $\epsilon \in (0,1]$,
where $\{|{\perp_j}\>:j=1,\ldots,N\}$ are orthonormal states satisfying
\ba
 \<\psi_j|{\perp_k}\>=\<\psi_j|S|{\perp_k}\>=\<{\perp_j}|S|{\perp_k}\>=0
 \text{~for~} j,k=1,\ldots,N.
\label{eq:perpstates}
\ea
Correspondingly, we modify the isometry $T$ to $T_\epsilon := \sum_j|\psi_j^\epsilon\>\<j|$.
Using the fact that $\<\psi_j^\epsilon|H|\psi_k^\epsilon\> = \epsilon\<\psi_j|H|\psi_k\>$, it can be shown that Theorems~\ref{thm:spectrum} and \ref{thm:dynamics} still hold, but with $H$ replaced by $\epsilon H$, so that $\lambda$ is replaced by $\epsilon\lambda$ and $h$ is replaced by $\epsilon h$.
Note that it is not necessary for the states $|\psi_j^\epsilon\>$ to arise from some Hermitian operator as in \eq{jstates}.
Perhaps the simplest choice is to enlarge the Hilbert space from $\C^N \times \C^N$ to $\C^{N+1} \times \C^{N+1}$, and let $|{\perp_j}\> = |j,N+1\>$.

Overall, we obtain a procedure for simulating a continuous-time walk by a discrete-time one, as follows:
\begin{enumerate}[itemsep=0pt]
  \item Given an initial state $|\phi_0\> \in \spn\{|j\>: j=1,\ldots,N\}$, apply the isometry $T_\epsilon$ and the unitary operation $\frac{1+\ii S}{\sqrt{2}}$.
  \item Apply $\tau$ steps of the discrete-time quantum walk $U := \ii S(2T_\epsilon T_\epsilon^\dag-1)$.
  \item Project onto the basis of states $\{\frac{1+\ii S}{\sqrt 2}T_\epsilon|j\>: j=1,\ldots,N\}$.
\end{enumerate}
By \thm{dynamics} and the preceding discussion, the resulting outcomes are approximately distributed according to $\Pr(j) = |\<j|e^{-\ii \tau \epsilon H/\norm{\abs(H)}}|\phi_0\>|^2$ provided $\epsilon^2 h^2$ and $\epsilon^3 h^3 \tau$ are small.  To obtain a total evolution time $t$, we choose $\epsilon = \norm{\abs(H)}t/\tau$ (subject to the constraint $\epsilon \le 1$), where the accuracy can be improved by increasing the number of simulation steps $\tau$, hence decreasing $\epsilon$.

More precisely, suppose our goal is to simulate the continuous-time quantum walk according to $H$ for a total time $t$, obtaining a final state $|\phi_t\>$ satisfying $\norm{e^{-\ii H t}|\phi_0\>-|\phi_t\>} \le \delta$.  Then it suffices to take
\be
  \tau \ge
  \max\left\{\norm{H}t
             \sqrt{\frac{1+(\tfrac{\pi}{2}-1)\norm{H}t}{\delta}},
             \norm{\abs(H)}t\right\}.
\label{eq:simtime}
\ee
In other words, the complexity of the simulation is $O\big((\norm{H} t)^{3/2}/\sqrt\delta,\norm{\abs(H)}t\big)$.

Note that to implement the discrete-time quantum walk defined above, it must be possible to implement the isometry $T$ (and its inverse).  In particular, this requires computing ratios of nonzero components of the principal eigenvector $|d\>$, a global property of $\abs(H)$.  While this may be difficult in general, it is tractable in many cases of interest.  If the graph is regular, then the principal eigenvector is simply the uniform superposition; if it is non-regular but sufficiently structured, then it may be possible to compute $|d\>$ explicitly.  For some applications it may be sufficient to precompute the principal eigenvector, as in \cite{ACRSZ07,RS07}.

Alternatively, another option is to replace the states in \eq{jstates} by
\begin{align}
 |\psi_j^{\prime\epsilon}\>
 &:= \sqrt{\frac{\epsilon}{\norm{H}_1}} \sum_{k=1}^N \sqrt{H^*_{jk}} \, |j,k\> 
   + \sqrt{1-\epsilon\sum_{k=1}^N \frac{|H_{jk}|}{\norm{H}_1}}
     \, |{\perp_j}\>
\label{eq:jstates_abs}
\end{align}
for some $\epsilon \in (0,1]$, where $\norm{H}_1 := \max_j \sum_{k=1}^N |H_{jk}|$, and where the states $|{\perp_j}\>$ satisfy \eq{perpstates} with $\sum_{k=1}^N \sqrt{H^*_{jk}} |j,k\>$ playing the role of $|\psi_j\>$.  With this choice, no information about the principal eigenvector of $\abs(H)$ is needed to implement the isometry $T'_\epsilon := \sum_{j=1}^N |\psi_j^{\prime\epsilon}\>\<j|$, and we retain the necessary properties that $\<\psi_j^{\prime\epsilon}|\psi_k^{\prime\epsilon}\>=\delta_{j,k}$ and $\<\psi_j^{\prime\epsilon}|S|\psi_k^{\prime\epsilon}\> \propto H_{jk}$.  However, since the quantity $\norm{\abs(H)}$ is replaced by the maximum absolute column sum norm $\norm{H}_1\ge\norm{\abs(H)}$, the simulation may be less efficient (and in particular, still suffers from the sign problem discussed in \sec{signproblem}).

This simulation immediately shows that the exponential speedup by continuous-time quantum walk demonstrated in \cite{CCDFGS03} carries over to the discrete-time model.  (This is a case where the graph is slightly non-regular, but the principal eigenvector can be computed explicitly.)  On the other hand, the speedup demonstrated in \cite{CSV07} apparently does \emph{not} carry over, as $\norm{\abs(H)} \gg \norm{H}$ for the relevant Hamiltonian, and exponentially many steps of the discrete-time quantum walk would be required.  We will return to this issue in \sec{signproblem}.

The estimate \eq{simtime} may be overly pessimistic for some applications.  To emulate the dynamics of the entire Hamiltonian, we choose $\epsilon$ so that all eigenvalues are small.  But for algorithms that effectively work in a low-energy subspace (e.g., \cite{CG03,CG04}), or that only depend on a small spectral gap (e.g., \cite{ACRSZ07,RS07}), it may be feasible to use little or no rescaling.

\section{Simulating non-sparse Hamiltonians}
\label{sec:nonsparse}

We now turn to the problem of simulating Hamiltonian dynamics.
In this section, we review known results about the simulation of sparse Hamiltonians, introduce the problem of simulating Hamiltonians that are not necessarily sparse, and describe some preliminary results on the simulation of non-sparse Hamiltonians.

Given a description of some Hermitian matrix $H$, Hamiltonian simulation is the problem of implementing the unitary operator $e^{-\ii H t}$ by a quantum circuit for any desired value of $t$.
Efficient simulations of sparse Hamiltonians are well-known.
In the special case that $H$ is local---i.e., when it is a sum of terms, each acting on a constant number of qubits---efficient simulation is straightforward \cite{Llo96}.
More generally, a Hamiltonian can be simulated efficiently provided it is \emph{sparse} and \emph{efficiently row-computable}.  We say a Hermitian operator $H$ acting on $\C^N$ is sparse in the basis $\{|j\>: j=1,\ldots,N\}$ if for any $j \in \{1,\ldots,N\}$, there are at most $\poly(\log N)$ values of $k \in \{1,\ldots,N\}$ for which $\<j|H|k\>$ is nonzero.  We say it is efficiently row-computable if there is an efficient procedure for determining those values of $k$ together with the corresponding matrix elements.

The simulability of sparse Hamiltonians was first explicitly stated in \cite{AT03}; it also follows from \cite{CCDFGS03} together with classical results on local coloring of graphs \cite{Lin92}.  The main idea is as follows.  Suppose we color the edges of the graph of nonzero entries of $H$; then the subgraphs of any particular color consist of isolated edges, meaning that the evolution on any one of these subgraphs takes place in isolated two-dimensional subspaces, and is easily simulated.  The subgraphs can be recombined using approximations to the Lie product formula,
\be
  \lim_{n \to \infty} \big(e^{-\ii A t/n} e^{-\ii B t/n}\big)^n = e^{-\ii (A+B) t}.
\ee
As the lowest-order approximation, supposing $\norm{A},\norm{B} \le h$, we have
\be
  \Norm{\big(e^{-\ii A t/n} e^{-\ii B t/n}\big)^n - e^{-\ii (A+B) t}} = O((ht)^2/n),
\ee
which shows that $n=O((ht)^2/\delta)$ simulation steps suffice to achieve error at most $\delta$.  Similarly, at second order,
\be
  \Norm{\big(e^{-\ii A t/2n} e^{-\ii B t/n} e^{-\ii A t/2n}\big)^n - e^{-\ii (A+B) t}} = O((ht)^3/n^2),
\ee
so that $n=O((ht)^{3/2}/\sqrt\delta)$ simulation steps suffice.
With approximations of increasingly high order, a simulation can be performed in nearly linear time \cite{BACS05,Chi04}.  Ultimately, we have

\begin{theorem} \label{thm:sparse}
Suppose the graph of nonzero entries of $H$ has $N$ vertices and maximum degree $d$, and that $H$ is efficiently row-computable, with $|\<j|H|k\>|\le h$ for all $j,k\in\{1,\ldots,N\}$.  Then evolution according to $H$ for time $t$ with error at most $\delta$ can be simulated in $ht(ht/\delta)^{o(1)} \cdot \poly(d,\log N)$ steps.
\end{theorem}

It is natural to ask under what conditions we can simulate a \emph{non-sparse} Hamiltonian.  Notice that if a Hamiltonian is not sparse, then it cannot be efficiently row-computable, simply because we cannot write down the nonzero entries of a row in polynomial time.  However, we can still suppose that the Hamiltonian has a succinct description of some kind and ask whether an efficient simulation is possible.

One way to simulate the dynamics of a non-sparse Hamiltonian is to use information about its spectrum.  By the simple identity $e^{-\ii H t} = U e^{-\ii U^\dag H U t} U^\dag$, we can simulate $H=UDU^\dag$ (with $D$ diagonal) provided we can efficiently perform a unitary transformation $U$ mapping the standard basis vector $|j\>$ to the $j$th eigenvector of $H$, and efficiently compute the $j$th eigenvalue of $H$, under some canonical ordering.  For example, this approach can be used when $H$ is the adjacency matrix of a complete graph, a complete bipartite graph, or a star graph (and even for some more complicated cases, e.g., the Winnie Li graph in odd dimensions \cite{CSV07}), all of which have high maximum degree.  However, it is unclear how broadly this strategy can be applied.

Here, we are interested in simulations based on the graph structure of the Hamiltonian in a fixed basis.
To begin, suppose the Hamiltonian is the adjacency matrix of an $N$-vertex graph $G$.  We say that such a Hamiltonian is \emph{efficiently index-computable} if, for any $j \in \{1,\ldots,N\}$, it is possible to efficiently compute $\deg(j)$ and to compute the $k$th neighbor of $j$ (under some canonical ordering) for any index $k \in \{1,\ldots,\deg(j)\}$.  The discrete-time quantum walk on an unweighted graph $G$ is straightforward to implement with such a description.  In particular, given a black box for the degree and the $j$th neighbors of any given vertex, a step of the discrete-time walk on $G$ can be simulated using only four queries (to compute the degree, compute a neighbor, uncompute the neighbor, and uncompute the degree).
More generally, we can consider graphs with edge weights under certain conditions mentioned below.

As a preliminary observation, we can obtain efficient simulations of some non-sparse Hamiltonians by generalizing the notion of edge coloring.  Suppose we have a decomposition of the graph of nonzero entries of the Hamiltonian into a union of polynomially many subgraphs, each of which is a disjoint union of simulable subgraphs; then we can efficiently simulate the dynamics using the Lie product formula.  \thm{sparse} can be viewed as the special case where the simulable subgraphs are single edges.  As a novel example of this technique, there is an efficient simulation of Hamiltonians whose graphs are trees:

\begin{theorem} \label{thm:trees}
Suppose the graph of nonzero entries of $H$ is a rooted, efficiently index-computable tree, with the parent of any vertex having a known index.  Furthermore, suppose that the weights on the edges from any given vertex are \emph{integrable}, in the sense of \cite{GR02} (which holds trivially for an unweighted graph), and that for any vertex we can efficiently compute the distance to the root.  Then the evolution according to $H$ for time $t$ can be simulated in $O((ht)^{1+o(1)})$ steps, where $h := \max_j \sqrt{\sum_k |\<k|H|j\>|^2}$. 
\end{theorem}

\begin{proof}
Consider a decomposition of the tree into two forests of stars.  The first forest contains the complete star around the root and stars centered at all vertices an even distance from the root, including all child vertices (but not the parent vertex) in each star.  The second forest contains the stars around vertices an odd distance from the root, again including all child vertices but not the parent vertex.  The result is a decomposition $H=H_1+H_2$, where the graphs of nonzero entries of $H_1$ and $H_2$ are both forests of stars.

A single star with center vertex $0$ and weight $w_i$ on the edge $(0,i)$ for $i=1,\ldots,k$ has exactly two nonzero eigenvalues, $\lambda_\pm := \pm\sqrt{|w_1|^2+\cdots+|w_k|^2}$, with corresponding normalized eigenvectors
\ba
  \frac{1}{\sqrt2} \bigg( |0\> + \frac{1}{\lambda_\pm} \sum_{i=1}^k w_i |i\> \bigg).
\ea
To transform from the standard basis to the eigenbasis of the star, we can perform the operation $|0\> \mapsto (|0\>+|1\>)/\sqrt{2}$, $|1\> \mapsto (|0\>-|1\>)/\sqrt{2}$, $|i\> \mapsto |i\>$ for $i=2,\ldots,k$, followed by any operation satisfying $|0\> \mapsto |0\>$ and $|1\> \mapsto \sum_i w_i |i\>/\lambda_+$; the latter can be implemented by the techniques of \cite{GR02}.

By applying this simulation for all disjoint stars in one of the forests, we can simulate the dynamics of either $H_1$ or $H_2$.  Finally, using a high-order approximation of the Lie product formula, we can simulate the full Hamiltonian $H$ with nearly linear overhead.
\end{proof}

It is not clear how to extend this strategy to general graphs.  However, we will see next that the correspondence with discrete-time quantum walk gives a broadly applicable simulation method.

\section{Linear-time Hamiltonian simulation}
\label{sec:linear}

We now present the main result of the paper, a method of simulating Hamiltonian dynamics that exploits the connection between continuous- and discrete-time quantum walk.
One such approach is to simply apply \thm{dynamics}, which states that any Hamiltonian dynamics can be well-approximated by a corresponding discrete-time process.  However, with that approach, the number of steps used to simulate evolution for time $t$ scales as $t^{3/2}$.
Instead, we can obtain a linear-time simulation by combining \thm{spectrum} with phase estimation.  By the no fast-forwarding theorem \cite{BACS05}, this scaling is optimal, even for sparse Hamiltonians.  However, we do not require the Hamiltonian to be sparse.

The simulation of an $N \times N$ Hamiltonian $H$ proceeds as follows.
Given an input state
\be
  |\psi\> = \sum_\lambda \psi_\lambda |\lambda\>
  \label{eq:inputstate}
\ee
in $\C^N$, where $|\lambda\>$ denotes an eigenvector of $H/\norm{\abs(H)}$ with eigenvalue $\lambda$, apply the isometry $T$ defined in \eq{isometry} to create the state $T|\psi\> \in \C^N \otimes \C^N$.  This state may be written
\ba
  T|\psi\> &= \sum_\lambda \psi_\lambda T|\lambda\> \\
  &= \sum_\lambda \psi_\lambda \Big( \textstyle
      \frac{1-\lambda e^{-\ii\arccos\lambda}}{\sqrt{2(1-\lambda^2)}}|\mu_+\>
     +\frac{1-\lambda e^{ \ii\arccos\lambda}}{\sqrt{2(1-\lambda^2)}}|\mu_-\>
     \Big),
     \label{eq:inputstateisom}
\ea
where we have used \eq{discreteqwalkevec} to write each $T|\lambda\>$ in terms of the corresponding eigenvectors $|\mu_\pm\>$ of the discrete-time quantum walk $U$ corresponding to $H$.

Our goal is to introduce a phase $e^{-\ii\lambda t}$ for the $\lambda$ term of this superposition.  To this end, perform coherent phase estimation on the discrete-time quantum walk $U$ corresponding to $H$.  Recall from \thm{spectrum} that the eigenvalues of $U$ corresponding to the eigenvalue $\lambda$ of $H/\norm{\abs{H}}$ are $\pm e^{\pm \ii\arcsin\lambda}$
Of course, we are not directly interested in $\arcsin(\lambda)$, but it implicitly determines $\lambda$.
Given an estimate $\tilde\lambda \approx \lambda$, we induce the phase $e^{-\ii \tilde\lambda t}$, uncompute $\tilde\lambda$ by performing phase estimation in reverse, and finally apply the inverse isometry $T^\dag$.
Overall, we claim that this procedure achieves the following:

\begin{theorem}\label{thm:linearsim}
  A Hamiltonian $H$ can be simulated for time $t$ with fidelity at least $1-\delta$ using $O(\norm{\abs(H)}t/\sqrt\delta)$ steps of the corresponding discrete-time quantum walk \eq{discreteqwalk}.
\end{theorem}

\begin{proof}
Let $|\theta\>$ denote an eigenstate of the discrete-time quantum walk with eigenvalue $e^{\ii\theta}$.  Let $P$ denote the isometry that performs phase estimation on this walk, appending a register with an estimate of the phase, as follows:
\ba
  P|\theta\> = \sum_\phi \amp{\phi}{\theta} |\theta,\phi\>.
\ea
Here $\amp{\phi}{\theta}$ is the amplitude for the estimate $\phi$ when the eigenvalue is in fact $\theta$; in particular, $\Pr(\phi|\theta)=|\amp{\phi}{\theta}|^2$.  Let $F_t$ be the unitary operation that applies the desired phase, namely
\ba
  F_t|\theta,\phi\> = e^{-\ii t \sin\phi}|\theta,\phi\>.
\ea
Then our simulation of the Hamiltonian evolution $e^{-\ii H t}$ is $T^\dag P^\dag F_tPT$.

Given an input state $|\psi\>$ as in \eq{inputstate}, we compute the inner product between the ideal state $e^{-\ii H t} |\psi\>$ and the simulated state $T^\dag P^\dag F_t P T|\psi\>$.  Similarly to \eq{inputstateisom}, we have
\ba
  P T e^{-\ii H t} |\lambda\>
  &= e^{-\ii \lambda t} P T |\lambda\> \\
  &= e^{-\ii \lambda t} P \Big( \textstyle
      \frac{1-\lambda e^{-\ii\arccos\lambda}}{\sqrt{2(1-\lambda^2)}}|\mu_+\>
     +\frac{1-\lambda e^{ \ii\arccos\lambda}}{\sqrt{2(1-\lambda^2)}}|\mu_-\> 
     \Big) \\
  &\begin{aligned}= e^{-\ii \lambda t} \sum_\phi
     \!\!\!\!\!\!\begin{aligned}[t]
     \Big(\amp{\phi}{\arcsin\lambda} \textstyle
      \frac{1-\lambda e^{-\ii\arccos\lambda}}{\sqrt{2(1-\lambda^2)}}
      &|\mu_+,\phi\> \\
           +\, \amp{\phi}{\pi-\arcsin\lambda} \textstyle
      \frac{1-\lambda e^{ \ii\arccos\lambda}}{\sqrt{2(1-\lambda^2)}}
      &|\mu_-,\phi\> \Big),
    \end{aligned} \end{aligned}
\ea
and
\ba
  \begin{aligned}F_t P T |\lambda\>
    = \sum_\phi e^{-\ii t \sin\phi}
     \!\!\!\!\!\!\begin{aligned}[t]
     \Big(\amp{\phi}{\arcsin\lambda} \textstyle
      \frac{1-\lambda e^{-\ii\arccos\lambda}}{\sqrt{2(1-\lambda^2)}}
      &|\mu_+,\phi\> \\
        +\, \amp{\phi}{\pi-\arcsin\lambda} \textstyle
      \frac{1-\lambda e^{ \ii\arccos\lambda}}{\sqrt{2(1-\lambda^2)}}
      &|\mu_-,\phi\> \Big).
    \end{aligned} \end{aligned}
\ea
Therefore, using orthonormality of the eigenstates, we have
\ba
  \<\psi|e^{\ii H t} T^\dag P^\dag F_t P T|\psi\>
  &\begin{aligned}[t]= \sum_\lambda |\psi_\lambda|^2 \<\lambda|e^{\ii H t} T^\dag P^\dag F_t P T|\lambda\>\end{aligned} \\
  &\begin{aligned}
   = \sum_{\lambda,\phi} |\psi_\lambda|^2 e^{\ii(\lambda-\sin\phi)t}
     \!\!\!\!\!\!\!\!\begin{aligned}[t]
     \Big(
     |\amp{\phi}{\arcsin\lambda}|^2 & \Big| \textstyle
     \frac{1-\lambda e^{-\ii\arccos\lambda}}{\sqrt{2(1-\lambda^2)}}
     \Big|^2 \\
   +|\amp{\phi}{\pi-\arcsin\lambda}|^2 & \Big| \textstyle
     \frac{1-\lambda e^{\ii\arccos\lambda}}{\sqrt{2(1-\lambda^2)}}
     \Big|^2 \Big) \end{aligned} \end{aligned} \\
  &\begin{aligned}[t]= \sum_{\lambda,\phi} |\psi_\lambda|^2 e^{\ii(\lambda-\sin\phi)t}
     \frac{|\amp{\phi}{\arcsin\lambda}|^2 
          +|\amp{\phi}{\pi-\arcsin\lambda}|^2}{2}.
    \end{aligned}
\ea
Thus the fidelity of the simulation is
\ba
  |\<\psi|e^{\ii H t} T^\dag P^\dag F_t P T|\psi\>|
  &\ge \min_\lambda \sum_\phi \cos\big((\lambda-\sin\phi)t\big)
     \frac{|\amp{\phi}{\arcsin\lambda}|^2 
          +|\amp{\phi}{\pi-\arcsin\lambda}|^2}{2} \\
  &\ge \min_{\theta \in [0,2\pi)} \sum_\phi \cos\big((\sin\theta-\sin\phi)t\big)
     |\amp{\phi}{\theta}|^2 \\
  &\ge 1-\min_{\theta \in [0,2\pi)} \frac{1}{2} \sum_\phi
       \big((\sin\theta-\sin\phi)t\big)^2 |\amp{\phi}{\theta}|^2 \\
  &\ge 1-\frac{t^2}{2}\min_{\theta \in [0,2\pi)} \sum_\phi
       (\theta-\phi)^2 |\amp{\phi}{\theta}|^2,
\label{eq:simfidelity}
\ea
where in the last step we have used the fact that $|{\sin\theta-\sin\phi}| \le |{\theta-\phi}|$.

Now, to show that the fidelity is close to $1$, we quantify the performance of phase estimation.  To optimize the bound, we must choose the phase estimation procedure carefully.  In standard phase estimation modulo $M$ (using $M-1$ calls to the unitary operation whose phase is being estimated), we begin with the uniform superposition $\frac{1}{\sqrt M}\sum_{x=0}^{M-1}|x\>$ in the register used to estimate the phase.
However, other input states can give better estimates depending on the application.  In particular, the initial state
\ba
  \sqrt{\frac{2}{M+1}}\sum_{x=0}^{M-1} \sin\frac{\pi(x+1)}{M+1}|x\>
  \label{eq:pestvarstate}
\ea
minimizes the variance of the estimate \cite{BDM99,DDEMM07b,LP96}, so it gives the best possible bound in \eq{simfidelity}.
This initial state leads to the probability distribution
\ba
  |a_{\theta+\Delta|\theta}|^2
  &= \frac{\cos^2(\Delta\frac{M+1}{2}) \sin^2(\frac{\pi}{M+1})}
          {2M(M+1)\sin^2(\frac{\Delta}{2}+\frac{\pi}{2(M+1)})
                  \sin^2(\frac{\Delta}{2}-\frac{\pi}{2(M+1)})},
\ea
where the estimated phase is $\theta+\Delta = 2 \pi j/M$ for some integer $j$.  Of course, only the value of $j \bmod M$ is significant, and we can choose the range of angles so that $\Delta=\Delta_0 + 2\pi j/M$, where $0 \le \Delta_0 < 2\pi/M$ and where the integer $j$ satisfies $ -\ceil{M/2}+1 \le j \le \floor{M/2}$.  Then for sufficiently large $M$,
\ba
  |a_{\theta+\Delta|\theta}|^2
  &\le \frac{\pi^2}
            {2M^4 \sin^2(\frac{\Delta}{2}+\frac{\pi}{2(M+1)})
                  \sin^2(\frac{\Delta}{2}-\frac{\pi}{2(M+1)})} \\
  &\le \frac{128\pi^2}
            {M^4 \Delta^4 (1-\frac{\pi^2}{\Delta^2 (M+1)^2})^2} \\
  &\le \frac{512\pi^2}{9M^4\Delta^4},
\ea
where the last step assumes that $\Delta \ge 2\pi/M$.
Since $2\pi j/M \le \Delta \le 2\pi(j+1)/M$,
\ba
  \sum_\phi (\theta-\phi)^2 |a_{\phi|\theta}|^2
  &\le \frac{4\pi^2}{M^2}+2\sum_{j=1}^\infty \frac{128(j+1)^2}{9M^2j^4} \\
  &= \frac{4\pi^2}{M^2}+\frac{256}{9M^2}\bigg(\frac{15\pi^2 + \pi^4 + 180\zeta(3)}{90} \bigg)\\
  &\le \frac{186}{M^2}.
\ea
Using this bound in \eq{simfidelity}, we obtain a simulation fidelity of at least $1-93t^2/M^2$.  Thus we find that $M=O(t/\sqrt\delta)$ steps suffice to obtain fidelity at least $1-\delta$.
\end{proof}

Of course, to carry out the simulation described in \thm{linearsim}, we must be able to implement the discrete-time quantum walk.
As discussed at the end of \sec{limit}, this may be difficult in general due to the dependence of the states \eq{jstates} on the principal eigenvector of $\abs(H)$.  However, it is straightforward to implement the walk for many cases of interest, such as for an unweighted regular graph.  Indeed, the discrete-time quantum walk can be carried out efficiently provided only that the Hamiltonian is efficiently index-computable, and that the weights appearing in the states \eq{jstates} are integrable in the sense of \cite{GR02}.  Alternatively, to avoid introducing a dependence on a principal eigenvector, one can use \eq{jstates_abs} at the expense of replacing $\norm{\abs(H)}$ by $\norm{H}_1$.

For Hamiltonians whose graphs are trees, the simulations of Theorems \ref{thm:trees} and \ref{thm:linearsim} are incomparable.  On the one hand, the simulation of \thm{trees} runs in slightly superlinear time.  On the other hand, the simulation of \thm{linearsim} scales with $\norm{\abs(H)}$,
whereas \thm{trees} only scales with $\max_j \sqrt{\sum_k |\<k|H|j\>|^2}$, which may be smaller.

\section{Element distinctness}
\label{sec:ed}

This section gives an application of \thm{linearsim} to quantum query complexity: a continuous-time quantum walk algorithm for the element distinctness problem.  The intention is to formulate and analyze the algorithm entirely in terms of its Hamiltonian, but to ultimately quantify its complexity in terms of conventional quantum queries.  The relationship of this algorithm to Hamiltonian-based models of query complexity is discussed at the end of \sec{hamquery}.

In the element distinctness problem, we are given a black-box function $f:\{1,\ldots,N\}\to S$ (for some finite set $S$) and are asked to determine whether there are two indices $x,y \in \{1,\ldots,N\}$ such that $f(x)=f(y)$.  Ambainis found a discrete-time quantum walk algorithm that solves this problem with $O(N^{2/3})$ queries \cite{Amb07}, which is optimal \cite{AS04}.  Unlike quantum algorithms for search on low-degree graphs \cite{AKR04,CG03,CG04,Tul08}, no continuous-time analog of the element distinctness algorithm has been known: the walk takes place on a high-degree Johnson graph, and sparse Hamiltonian techniques are insufficient to implement it.  Since Ambainis's algorithm appeared, it has been an open question to find a continuous-time version (see for example \cite[Section 5]{CE03}).  We now describe a continuous-time quantum walk that, when simulated using \thm{linearsim}, gives an $O(N^{2/3})$-query quantum algorithm for element distinctness.

As in \cite{Amb07}, the algorithm uses a walk on the Johnson graph $J(N,M)$.  The vertices of this graph are the $\binom{N}{M}$ subsets of $\{1,\ldots,N\}$ of size $M$; edges connect subsets that differ in exactly one element.  Let $M:=\nint{N^{2/3}}$, the nearest integer to $N^{2/3}$.

To simplify the analysis, suppose there is a unique pair of indices $x,y$ for which $f(x)=f(y)$.  By a classical reduction, this assumption is without loss of generality \cite{Amb07}.
Let $A_j$ denote the set of $M$-element subsets of $\{1,\ldots,N\}$ that contain $j$ elements from $\{x,y\}$.  We use the convention that a set $S$ written in a ket denotes the uniform superposition $|S\> := \sum_{s \in S} |s\>/\sqrt{|S|}$ over the elements of $S$.  Then in the basis $\{|A_0\>,|A_1\>,|A_2\>\}$, the adjacency matrix of $J(N,M)$ is
\be
\begin{pmatrix}
  -2M & \sqrt{2M(N-M-1)} & 0 \\
  \sqrt{2M(N-M-1)} & 2-N & \sqrt{2(N-M)(M-1)} \\
  0 & \sqrt{2(N-M)(M-1)} & 2(M-N) \\
\end{pmatrix} + M(N-M).
\ee

We modify the graph to depend upon the black box as follows.  For every subset in $A_2$, add an extra vertex connected by an edge to the original subset.  Denote the set of all such extra vertices by $B_2$.  Then the Hamiltonian for the algorithm is $H = H_U + H_C$, where, in the basis $\{|A_0\>,|A_1\>,|A_2\>,|B_2\>\}$,
\ba
  H_U &=
  \frac{1}{N^{2/3}}
  \begin{pmatrix}
    -2M & \sqrt{2M(N-M-1)} & 0 & 0 \\
    \sqrt{2M(N-M-1)} & 2-N & \sqrt{2(N-M)(M-1)} & 0 \\
    0 & \sqrt{2(N-M)(M-1)} & 2(M-N) & 0 \\
    0 & 0 & 0 & 0
  \end{pmatrix} \\
  H_C &=
  \begin{pmatrix}
    0 & 0 & 0 & 0 \\
    0 & 0 & 0 & 0 \\
    0 & 0 & 0 & 1 \\
    0 & 0 & 1 & 0
  \end{pmatrix}.
\ea

For the initial state, take a uniform superposition over the vertices of the original Johnson graph, namely
\ba
  |\psi(0)\> &:= |A_0 \cup A_1 \cup A_2\> \\
  &= \Big(\sqrt{\tbinom{N-2}{M}}|A_0\>+\sqrt{2\tbinom{N-2}{M-1}}|A_1\>+\sqrt{\tbinom{N-2}{M-2}}|A_2\>\Big)/\sqrt{\tbinom{N}{M}} \label{eq:edstartstate} \\
 &= \big(1 - O(N^{-1/3})\big) |A_0\> + O(N^{-1/6}) |A_1\> + O(N^{-1/3}) |A_2\>
\ea
where we have used the choice $M=\nint{N^{2/3}}$ in the last line.  Asymptotically, the starting state is essentially $|A_0\>$.

To analyze the algorithm, we compute the spectrum of $H$.  In principle, this could be done using techniques from \cite{CG03,CG04}.
However, since the evolution takes place within a four-dimensional subspace, we can compute the relevant eigenvalues and eigenvectors in closed form.  In particular, we find two eigenvalues $\lambda_\pm = \frac{-1\pm\sqrt{17}}{4} N^{1/3} + O(1)$ with eigenvectors
\ba
  |\lambda_+\> &= \left(\sqrt{\mu}+O(\tfrac{1}{N^{1/3}}),O(\tfrac{1}{N^{1/6}}),O(\tfrac{1}{N^{1/2}}),\sqrt{1-\mu}+O(\tfrac{1}{N^{1/6}})\right) \nonumber\\
  &\approx (0.6154,0,0,0.7882) \\
  |\lambda_-\> &= \left(\sqrt{1-\mu}+O(\tfrac{1}{N^{1/3}}),O(\tfrac{1}{N^{1/6}}),O(\tfrac{1}{N^{1/2}}),-\sqrt{\mu}+O(\tfrac{1}{N^{1/6}})\right) \nonumber\\
  &\approx (0.7882,0,0,-0.6154),
\ea
where
\be
  \mu := \frac{8}{17+\sqrt{17}} \approx 0.3787.
\ee
Asymptotically, the algorithm is effectively confined to the two-dimensional subspace spanned by $|A_0\>$ and $|B_2\>$.   The state rotates within this subspace at a rate determined by the inverse of the gap $\lambda_+-\lambda_- = O(N^{-1/3})$.  Therefore, the initial state \eq{edstartstate} reaches a state with overlap $O(1)$ on $|B_2\>$ in time $O(N^{1/3})$.

To quantify the query complexity of this approach to element distinctness, we invoke \thm{linearsim}.  As in \cite{Amb07}, suppose we store the $M$ function values along with the subset at each vertex.  Then preparing \eq{edstartstate} takes $M=\nint{N^{2/3}}$ queries.  Furthermore, a step of the corresponding discrete-time quantum walk can be simulated using two queries: the walk operator is local on the Johnson graph, so we simply uncompute one function value and compute another.  Given the function values for each subset, the extra edges for marked vertices can be included without any additional queries; indeed, by an observation in the proof of Theorem 8 of \cite{ACRSZ07}, they can be implemented by performing the walk on the graph with extra edges for \emph{every} vertex of the Johnson graph, together with a phase factor for the edges corresponding to marked vertices.  Thus, the total number of queries used by the algorithm is $O(N^{2/3} + \norm{\abs(H)} N^{1/3})$.  Since $\norm{\abs(H)} = 2N^{1/3}+O(N^{-1/3})$, we find an algorithm with running time (and, in particular, query complexity) $O(N^{2/3})$.

\section{The continuous-time query model}
\label{sec:hamquery}

Continuous-time quantum walk has proven useful in motivating new algorithmic ideas.  Ultimately, though, it is most straightforward to quantify the complexity of the resulting algorithms by the number of elementary gates needed to simulate them in the quantum circuit model, as in the example of \sec{ed}.

However, one can also formulate a notion of query complexity directly in a Hamiltonian-based model of computation.  Such a model was first introduced by Farhi and Gutmann to describe a continuous-time analog \cite{FG96} of Grover's search algorithm \cite{Gro97}; it was later studied in a broader context by Mochon \cite{Moc06} and applied to the quantum walk algorithm for evaluating balanced binary game trees \cite{FGG07}.

In a general formulation of the continuous-time query model, an algorithm is described by a Hamiltonian of the form $H_D(t)+H_Q$, where $H_Q$ is a fixed, time-independent Hamiltonian encoding a black-box input, and $H_D(t)$ is an arbitrary oracle-independent ``driving Hamiltonian'' (possibly time-dependent and with no a priori upper bound on its norm).  The complexity is quantified simply by the total evolution time required to produce the result with bounded error.  Equivalently, we can consider a model of fractional queries interspersed by non-query unitary operations, and take the limit in which the fractional queries can be arbitrarily close to the identity, charging only $1/k$ of a full query to perform the $k$th root of a query.

It is clear that the continuous-time query model is at least as powerful as the conventional query model.  Very recently, it was shown that the continuous-time query model is in fact not significantly more powerful.  In particular, any algorithm using continuous queries for time $t$ can be simulated with $O(t \log t/\log\log t)$ discrete queries \cite{CGMSY08}.

Here we consider the case where the driving Hamiltonian $H_D(t)$ is restricted to be time-independent.  We suppose that the query Hamiltonian $H_Q$ has the form used in \cite{FGG07}: namely, for a binary black-box input $x \in \{0,1\}^N$, $H_Q$ acts on $\C^N \otimes \C^2$ as
\ba
  H_Q |i,b\> = |i,b \oplus x_i\>
  \label{eq:queryham}
\ea
for any $i \in \{1,\ldots,N\}$ and $b \in \{0,1\}$.  In other words, the query Hamiltonian describes a graph on $2N$ vertices, with an edge between vertices $(i,0)$ and $(i,1)$ if $x_i=1$, and no such edge if $x_i=0$.

Consider the case where $H_D$ is time-independent.
By simulating $H_D+H_Q$ with high-order approximations of the Lie product formula \cite{BACS05,Chi04}, the evolution for time $t$ can be approximated using $(\norm{H_D}t)^{1+o(1)}$ queries to a unitary black box for the input $x$.  Thus, a continuous-time quantum walk algorithm with $\norm{H_D}=O(1)$ gives rise to a conventional quantum query algorithm using only $t^{1+o(1)}$ queries.  However, this simulation incurs more overhead than that of \cite{CGMSY08}.

Instead, applying \thm{linearsim}, and again using the observation from \cite[Theorem 8]{ACRSZ07} to implement the discrete-time quantum walk using discrete queries (as in the simulation of the element distinctness algorithm in \sec{ed}), we find

\begin{theorem}\label{thm:hamquery}
Consider a continuous-time query algorithm with the time-independent Hamiltonian $H_D + H_Q$ that runs in time $t$.  Then the query complexity in the conventional discrete-time model is $O(\norm{\abs(H_D)}t)$.
\end{theorem}

In the case where $H_D$ is time-independent and satisfies $\norm{\abs(H_D)}=O(1)$, this simulation outperforms \cite{CGMSY08}.  However, it may be less efficient if $\norm{\abs(H_D)} \gg 1$, and is not even applicable to a general time-dependent $H_D(t)$.  Note, however, that almost all known continuous-time quantum algorithms have $H_D$ time-independent (with the notable exception of adiabatic algorithms for search \cite{DMV01,RC02}, which can nevertheless be simulated efficiently in the quantum circuit model \cite{RC03}) and satisfy $\norm{\abs(H_D)}=O(1)$.  (It might be interesting to investigate the application of discrete-time quantum walk to simulating the dynamics of a time-dependent Hamiltonian.)

Notice that the element distinctness algorithm from \sec{ed} does not naturally fit into the continuous-time query model.  The simulation described in \sec{ed} uses $O(N^{2/3})$ initial discrete queries, followed by $O(N^{2/3})$ queries to simulate an evolution for time only $O(N^{1/3})$.  However, the techniques of \cite{CGMSY08} could not be used to simulate the latter evolution using only $O(N^{1/3} \log N)$ discrete queries: the initial queries and the queries used to simulate the evolution should balance to avoid violating the $\Omega(N^{2/3})$ lower bound for element distinctness \cite{AS04}.  Despite the superficial similarity between $H_C$ and $H_Q$, it is the term $H_U$, rather than $H_C$, that depends on the black-box input; overall, $H_U + H_C$ does not have the form of an oracle-independent Hamiltonian plus a query Hamiltonian of the form \eq{queryham}.

\section{A sign problem for Hamiltonian simulation}
\label{sec:signproblem}

Although the Hamiltonian simulations described in Sections~\ref{sec:limit}--\ref{sec:linear} go considerably beyond previous techniques, they stop short of what might be possible.  We conclude by considering two problems for which improved simulation methods would be valuable: approximating exponential sums and implementing quantum transforms over association schemes.  Hopefully, these potential applications will motivate further work on the simulation of Hamiltonian dynamics.

The essential problem with simulations based on the correspondence to discrete-time quantum walk has to do with the appearance of $\norm{\abs(H)}$, rather than $\norm{H}$, in \thm{linearsim}.  The natural scaling parameter for Hamiltonian simulation would seem to be the basis-independent quantity $\norm{H}t$ rather than the basis-dependent quantity $\norm{\abs(H)}t$; it seems reasonable to attempt a simulation in time $\poly(\norm{H}t)$.  However, we do not currently know how to do this except in special cases.

There are at least two potential approaches to circumventing this limitation.  It might be possible to use decomposition techniques, such as the decomposition of trees into stars described in \sec{nonsparse}, in a more general context.  We also might try to perform simulations in alternative bases that can be reached by efficient unitary transformations.

\subsection{Approximating Kloosterman sums}

An \emph{exponential sum} over $\F_q$, the finite field with $q$ elements, is an expression of the form
\be
  \sum_{x \in \F_q} \chi(f(x)) \, \psi(g(x)),
\label{eq:expsum}
\ee
where $\chi$ and $\psi$ are multiplicative and additive characters of $\F_q$, respectively, and $f,g \in \F_q[x]$ are polynomials.  It is well known that, under fairly mild conditions,
\be
  \bigg|\sum_{x \in \F_q} \chi(f(x)) \, \psi(g(x))\bigg| \le \alpha \sqrt{q}
\ee
where the coefficient $\alpha$ depends on the degrees of $f$ and $g$ (see for example \cite[Theorem 2.6]{Sch04}).  However, computing the value of an exponential sum---or even approximating its magnitude as a fraction of $\sqrt{q}$ with, say, constant precision---appears to be a difficult problem in general.

When $f$ and $g$ are both linear in $x$, the sum \eq{expsum} is known as a \emph{Gauss sum}.  Gauss sums have magnitude precisely $\sqrt{q}$ provided $\chi,\psi$ are nontrivial.  No efficient classical algorithm for computing the phase of a general Gauss sum is known, but this phase can be efficiently approximated using a quantum computer \cite{DS03}.  It is an open question whether quantum computers can efficiently approximate other exponential sums.  (Note that for the special case of small characteristic, it is possible to calculate more general exponential sums using the results of \cite{Ked06}, as observed by Shparlinski; see \cite{CSV07}.)

Another type of exponential sum of particular interest is the \emph{Kloosterman sum}, which is obtained by letting $\chi$ be the quadratic character, $f(x)=x^2-c$ for some fixed $c \in \F_q$, and $g(x)=x$.  It is easy to see that such a sum is real-valued.  A classic result of Weil says that, provided $\psi$ is nontrivial and $c \ne 0$, the absolute value of the Kloosterman sum is at most $2\sqrt{q}$ \cite{Wei48}.  However, I am not aware of an efficient algorithm to compute even a single nontrivial bit of information about the Kloosterman sum, such as its sign or whether its magnitude is larger than $\sqrt{q}$.

Kloosterman sums arise naturally in problems involving hidden nonlinear structures over finite fields \cite{CSV07}.  In particular, the eigenvalues of a graph known as the \emph{Winnie Li graph} (in even dimensions) are proportional to Kloosterman sums.  An efficient quantum algorithm for approximating Kloosterman sums would provide new efficient quantum algorithms for certain hidden nonlinear structure problems.  For example, this would give a way to implement a quantum walk that could be used to solve the so-called hidden flat of centers problem.  Conversely, an implementation of that quantum walk could be used to estimate the sums, using phase estimation.

Let us consider a simple variant of the Winnie Li graph that exemplifies the relevant problem.  Let $G$ be the Cayley graph of the additive group of $\F_q$ with the generating set $X:=\{x \in \F_q: \chi(x^2 - c)=+1 \}$ for some $c \in \F_q^\times$, where $\chi$ denotes the quadratic character of $\F_q$.  Let $p$ be the characteristic of $\F_q$, and let $\tr:\F_q \to \F_p$ be the trace map, defined by $\tr(x)=x+x^p+x^{p^2}+\cdots+x^{q/p}$.  Observe that
\be
  \delta[x \in X] = \frac{1}{2}\big(1+\chi(x^2-c)-\delta[x^2=c]\big),
\ee
where $\delta[P]$ is $1$ if $P$ is true and $0$ if $P$ is false.
Then for each $k \in \F_q$, the Fourier vector
\ba
  |\tilde k\> := \frac{1}{\sqrt q}\sum_{x \in \F_q} \omega_p^{\tr(kx)}|x\>
  \label{eq:fourierstates}
\ea
(where $\omega_p := e^{2\pi\ii/p}$) is an eigenvector of $G$ with eigenvalue
\ba
  \sum_{x \in X} \omega_p^{\tr(kx)}
  &= \sum_{x \in \F_q} \frac{1}{2}\big(1+\chi(x^2-c)-\delta[x^2=c]\big) \omega_p^{\tr(kx)} \\
  &= \frac{1}{2}\bigg(q \, \delta[k=0]+\sum_{x \in \F_q} \chi(x^2-c) \, \omega_p^{\tr(kx)}\bigg) - \cos\frac{2\pi k \sqrt{c}}{p}
\ea
(where the term involving $\sqrt{c}$ does not appear if $c$ is not a square in $\F_q$).  In particular, the eigenvalues of $G$ are simply related to Kloosterman sums.

Notice that $k=0$ gives an eigenvalue $|X| = (q + 1)/2 - \delta[\chi(c)=+1] = (q \pm 1)/2$ (the degree of $G$) corresponding to the uniform vector.  Since we are only interested in the nontrivial Kloosterman sums, we can subtract off the projection of the adjacency matrix onto the uniform vector, defining a symmetric matrix $H$ with
$
  \<x|H|x'\> = \delta[x-x' \in X] - {|X|}/{q}
$
for
$
  x,x' \in \F_q
$.
This Hamiltonian satisfies $\norm{H} \le 2\sqrt{q}$ (by Weil's Theorem), and a simulation of its dynamics for time $t$ in $\poly(\norm{H}t,\log q)$ steps would give an efficient quantum algorithm for approximating Kloosterman sums with constant precision.

Unfortunately, \thm{linearsim} is insufficient for this task.  Whereas $\norm{H} \sim 2\sqrt{q}$, a simple calculation shows that $\norm{\abs(H)} = (q^2-1)/2q = \Theta(q)$, so $\norm{\abs(H)}/\norm{H}$ is exponentially large (in $\log q$).  Similar considerations hold for the actual Winnie Li graph.

\subsection{Quantum transforms for association schemes}

An \emph{association scheme} is a combinatorial object with useful algebraic properties (see \cite{God05} for an accessible introduction).  We say that a set of matrices $A_0,A_1,\ldots,A_d \in \{0,1\}^{N \times N}$ is a $d$-class association scheme on $N$ vertices provided (i) $A_0 = I$, the $N \times N$ identity matrix; (ii) $\sum_{i=0}^d A_i = J$, the $N \times N$ matrix with every entry equal to $1$; (iii) $A_i^T \in \{A_0,\ldots,A_d\}$ for each $i \in \{0,\ldots,d\}$; and (iv) $A_i A_j = A_j A_i \in \spn\{A_0,\ldots,A_d\}$ for each $i,j \in \{0,\ldots,d\}$.  An association scheme partitions the relationships between pairs of vertices into classes: we say that the relationship between vertex $x$ and vertex $y$ is of type $i$ if $(A_i)_{x,y}=1$ (which implies that $(A_j)_{x,y}=0$ for all $j \ne i$).  Let $N_i$ denote the number of vertices of type $i$ relative to any given vertex (or equivalently, the number of $1$s in any given row of $A_i$).

The matrices $A_0,A_1,\ldots,A_d$ generate a $(d+1)$-dimensional algebra over $\C$ called the \emph{Bose-Mesner algebra} of the scheme.  This algebra is also generated by a set of $d+1$ projection matrices $E_0,E_1,\ldots,E_d$ called its \emph{idempotents}.  These projections are orthonormal ($E_i E_j = \delta_{i,j} E_i$ for all $i,j \in \{0,\ldots,d\}$) and complete ($\sum_{i=0}^d E_i = I$); furthermore, the range of $E_i$ is an eigenspace of $A_j$ for every $i,j \in \{0,\ldots,d\}$.  Since they lie in the Bose-Mesner algebra, the idempotents can be expanded as $E_i = \frac{1}{N} \sum_{j=0}^d q_{ij} A_j$, where the $q_{ij}$ are referred to as the \emph{dual eigenvalues} of the scheme.

The unitary matrices in the Bose-Mesner algebra are precisely those matrices of the form
\be
  U = \sum_{i=0}^d e^{\ii\phi_i} E_i
  \label{eq:assocunitary}
\ee
for some phases $\phi_i \in \R$.  Since they are highly structured, these unitary operators can be concisely specified even when we think of $N$ as exponentially large, and thus they represent a natural class of quantum transforms for possible use in quantum algorithms.  Closely related to such a transform is the POVM with elements $\{E_0,E_1,\ldots,E_d\}$: the ability to perform this measurement (coherently) can be used to implement any desired unitary transform in the Bose-Mesner algebra, and conversely, implementation of any unitary transform in the Bose-Mesner algebra with well-separated eigenvalues can be used to measure $\{E_0,E_1,\ldots,E_d\}$, by phase estimation.

For a concrete application of association scheme transforms, consider the following problem.  Fix a symmetric $N$-vertex association scheme and an (unknown) vertex $t$ of that scheme.  Suppose we are given a black box function $f:\{1,\ldots,N\} \to S$ with the promise that for all $x,y \in \{1,\ldots,N\}$, $f(x)=f(y)$ if and only if $x$ and $y$ are of the same class with respect to $t$ (or in other words, if and only if $(A_i)_{t,x} = (A_j)_{t,y} = 1 \implies i=j$).  The task is to learn $t$ using as few queries to $f$ as possible, ideally only $\poly(\log N)$.

The \emph{shifted quadratic character problem} is a particular instance of this problem in which the association scheme is the \emph{Paley scheme}.  It has $d=2$ classes, meaning that it corresponds to a \emph{strongly regular graph}, the Paley graph.  In the Paley scheme, the vertices are elements of $\F_q$ (where $q=1 \bmod 4$), and $(A_1)_{x,y}=1$ if and only if $x-y$ is a square in $\F_q$.  The sequence of values of the quadratic character $\chi(t), \chi(t+1), \chi(t+2), \ldots$ has been proposed as a pseudorandom generator \cite{Dam88}, and indeed no efficient classical algorithm is known that will learn $t$ from the function $f(x)=\chi(t+x)$ that hides $t$ in the Paley scheme.  On the other hand, van Dam, Hallgren, and Ip discovered a quantum algorithm that uses the quantum Fourier transform to solve this problem in time $\poly(\log q)$ \cite{DHI03}.

Association scheme transforms provide a general approach to this problem.  Suppose we prepare a uniform superposition over all vertices of the scheme, compute the hiding function $f$, and then discard its value.  The result is a uniform superposition over points that are of type $i$ relative to $t$,
\be
  |\psi_i\> := \frac{1}{\sqrt{N_i}} \sum_{x:(A_i)_{t,x}=1} |x\>,
\ee
where type $i$ occurs with probability $N_i/N$.  Now suppose we apply the unitary operator \eq{assocunitary}, giving the state $U|\psi_i\>$, and measure the vertex.  The probability that we obtain the vertex $t$ is
\ba
  |\<\psi_0|U|\psi_i\>|^2
  = \frac{1}{N^2} \bigg| \sum_{j,k=0}^d e^{\ii \phi_j}
     q_{jk} \<\psi_0| A_k |\psi_i\> \bigg|^2
  = \frac{N_i}{N^2} \bigg| \sum_{j=0}^d e^{\ii \phi_j} q_{ji} \bigg|^2.
\ea
Supposing that we know the value of $i$, we can choose the phases $\phi_j$ to obtain a success probability $N_i(\sum_j |q_{ji}|)^2/N^2$.  Since type $i$ occurs with probability $N_i/N$, by always choosing the phases to optimize the probability of obtaining vertex $t$ given that type $i$ occurred (and pessimistically assuming that we never obtain vertex $t$ when another type occurs), we can succeed with probability at least $N_i^2(\sum_j |q_{ji}|)^2/N^3$.  Thus, by always choosing the phases according to the value of $i$ that maximizes this expression, we can achieve an overall success probability of at least
\be
  \max_i \frac{N_i^2}{N^3} \bigg( \sum_{j=0}^d |q_{ji}| \bigg)^2.
\ee
For example, in the Paley scheme, a straightforward calculation shows that this approach succeeds with probability at least $(\sqrt{q}+1)^2(q-1)^2/4q^3 = 1/4 + O(1/\sqrt{q})$.

The unitary operators \eq{assocunitary} for the Paley scheme can be implemented efficiently because the corresponding idempotents are projections onto Fourier basis vectors.  In particular, $E_0 = |\tilde 0\>\<\tilde 0|$, $E_1$ projects onto $\spn\{|\tilde k\>: \chi(k)=+1\}$, and $E_2$ projects onto $\spn\{|\tilde k\>: \chi(k)=-1\}$ (recall \eq{fourierstates}).  Thus, using the Fourier transform and the ability to compute quadratic characters (which can be done efficiently using Shor's algorithm for discrete logarithms \cite{Sho97}), we can efficiently carry out the above strategy to solve the shifted quadratic character problem.

However, this approach is not very different from the one in \cite{DHI03}, which employs similar tools.  Instead, it would be appealing to have a purely combinatorial means of implementing \eq{assocunitary}, which could be applied even for association schemes without an underlying group structure.

One way to implement \eq{assocunitary} would be to coherently measure each of the idempotents in turn, performing a phase shift conditional on the measurement outcome and then undoing the measurement.  In principle, these measurements could be performed by applying phase estimation to a simulation of Hamiltonian dynamics with the Hamiltonian given by the idempotent.  Unfortunately, in many cases of interest, this approach encounters the same sign problem seen in the previous section.  For example, whereas $\norm{E_1}=1$, one can show that $\norm{\abs(E_1)}=(\sqrt{q}+1)(q-1)/2q=\Theta(\sqrt{q})$.  Since $\norm{\abs(E_1)}/\norm{E_1}$ is exponentially large in $\log q$, we again find that \thm{linearsim} does not give an efficient implementation.

\section*{Acknowledgements}

I would like to thank Ben Reichardt and Rolando Somma for pointing out the choice \eq{jstates_abs}; Ben Reichardt, Robert \v{S}palek, and Shengyu Zhang for discussions, in the course of writing \cite{ACRSZ07}, that led to the simulation described in \thm{trees}; Chris Godsil and Simone Severini for discussions of association schemes; and Dominic Berry and Aram Harrow for discussions of phase estimation.
This work was supported by MITACS, NSERC, QuantumWorks, and the US ARO/DTO.


\newcommand{\noopsort}[1]{}
\providecommand{\bysame}{\leavevmode\hbox to3em{\hrulefill}\thinspace}


\begin{thebibliography}{10}
\setlength{\itemsep}{0pt}

\bibitem{AS04}
S.~Aaronson and Y.~Shi, \emph{Quantum lower bounds for the collision and the
  element distinctness problems}, J. ACM \textbf{51} (2004), no.~4, 595--605,
  \eprint{quant-ph/0111102} and \eprint{quant-ph/0112086}, preliminary versions
  in STOC 2002 and FOCS 2002.

\bibitem{AAKV01}
D.~Aharonov, A.~Ambainis, J.~Kempe, and U.~Vazirani, \emph{Quantum walks on
  graphs}, Proc. 33rd ACM Symposium on Theory of Computing, pp.~50--59, 2001,
  \eprint{quant-ph/0012090}.

\bibitem{AT03}
D.~Aharonov and A.~Ta-Shma, \emph{Adiabatic quantum state generation and
  statistical zero knowledge}, Proc. 35th ACM Symposium on Theory of Computing,
  pp.~20--29, 2003, \eprint{quant-ph/0301023}.

\bibitem{AF02}
D.~Aldous and J.~A. Fill, \emph{Reversible {M}arkov Chains and Random Walks on
  Graphs}, in preparation,
  \url{http://www.stat.berkeley.edu/~aldous/RWG/book.html}.

\bibitem{Amb07}
A.~Ambainis, \emph{Quantum walk algorithm for element distinctness}, SIAM J.
  Comput. \textbf{37} (2007), no.~1, 210--239, \eprint{quant-ph/0311001},
  preliminary version in FOCS 2004.

\bibitem{ABNVW01}
A.~Ambainis, E.~Bach, A.~Nayak, A.~Vishwanath, and J.~Watrous,
  \emph{One-dimensional quantum walks}, Proc. 33rd ACM Symposium on Theory of
  Computing, pp.~37--49, 2001, \eprint{quant-ph/0010117}.

\bibitem{ACRSZ07}
A.~Ambainis, A.~M. Childs, B.~W. Reichardt, R.~{\v S}palek, and S.~Zhang,
  \emph{Any {AND-OR} formula of size {$N$} can be evaluated in time {$N^{1/2 +
  o(1)}$} on a quantum computer}, Proc. 48th IEEE Symposium on Foundations of
  Computer Science, pp.~363--372, 2007, \eprint{quant-ph/0703015} and
  \eprint{arXiv:0704.3628}.

\bibitem{AKR04}
A.~Ambainis, J.~Kempe, and A.~Rivosh, \emph{Coins make quantum walks faster},
  Proc. 16th ACM-SIAM Symposium on Discrete Algorithms, pp.~1099--1108, 2005,
  \eprint{quant-ph/0402107}.

\bibitem{BACS05}
D.~W. Berry, G.~Ahokas, R.~Cleve, and B.~C. Sanders, \emph{Efficient quantum
  algorithms for simulating sparse {H}amiltonians}, Commun. Math. Phys.
  \textbf{270} (2007), no.~2, 359--371, \eprint{quant-ph/0508139}.

\bibitem{BS06}
H.~Buhrman and R.~\v{S}palek, \emph{Quantum verification of matrix products},
  Proc. 17th ACM-SIAM Symposium on Discrete Algorithms, pp.~880--889, 2006,
  \eprint{quant-ph/0409035}.

\bibitem{BDM99}
V.~Bu{\v{z}}ek, R.~Derka, and S.~Massar, \emph{Optimal quantum clocks},
  Phys. Rev. Lett. \textbf{82} (1999), no.~10, 2207--2210,
  \eprint{quant-ph/9808042}.

\bibitem{Chi04}
A.~M. Childs, \emph{Quantum information processing in continuous time}, Ph.D.
  thesis, Massachusetts Institute of Technology, Cambridge, MA, 2004.

\bibitem{CCDFGS03}
A.~M. Childs, R.~Cleve, E.~Deotto, E.~Farhi, S.~Gutmann, and D.~A. Spielman,
  \emph{Exponential algorithmic speedup by quantum walk}, Proc. 35th\ ACM
  Symposium on Theory of Computing, pp.~59--68, 2003, \eprint{quant-ph/0209131}.

\bibitem{CE03}
A.~M. Childs and J.~M. Eisenberg, \emph{Quantum algorithms for subset finding},
  Quantum Information and Computation \textbf{5} (2005), no.~7, 593--604,
  \eprint{quant-ph/0311038}.

\bibitem{CFG02}
A.~M. Childs, E.~Farhi, and S.~Gutmann, \emph{An example of the difference
  between quantum and classical random walks}, Quantum Information Processing
  \textbf{1} (2002), nos.~1-2, 35--43, \eprint{quant-ph/0103020}.

\bibitem{CG03}
A.~M. Childs and J.~Goldstone, \emph{Spatial search by quantum walk}, Phys.
  Rev. A \textbf{70} ({\noopsort{a}}2004), no.~2, 022314,
  \eprint{quant-ph/0306054}.

\bibitem{CG04}
\bysame, \emph{Spatial search and the {D}irac equation}, Phys. Rev. A
  \textbf{70} ({\noopsort{b}}2004), no.~4, 042312, \eprint{quant-ph/0405120}.

\bibitem{CSV07}
A.~M. Childs, L.~J. Schulman, and U.~V. Vazirani, \emph{Quantum algorithms for
  hidden nonlinear structures}, Proc. 48th\ IEEE Symposium on Foundations of
  Computer Science, pp.~395--404, 2007, \eprint{arXiv:0705.2784}.

\bibitem{CGMSY08}
R.~Cleve, D.~Gottesman, M.~Mosca, R.~D. Somma, and D.~L. Yonge-Mallo,
  \emph{Efficient discrete-time simulations of continuous-time quantum query
  algorithms}, Proc. 41st ACM Symposium on Theory of Computing, pp.~409--416,
  2009, \eprint{arXiv:0811.4428}.

\bibitem{DDEMM07b}
W. van Dam, G.~M. D'Ariano, A.~Ekert, C.~Macchiavello, and M.~Mosca,
  \emph{Optimal phase estimation in quantum networks},
  J. Phys. A \textbf{40} (2007), no.~28, 7971--7984, \eprint{arXiv:0706.4412}.

\bibitem{DHI03}
W.~van Dam, S.~Hallgren, and L.~Ip, \emph{Quantum algorithms for some hidden
  shift problems}, Proc. 14th ACM-SIAM Symposium on Discrete Algorithms,
  pp.~489--498, 2002, \eprint{quant-ph/0211140}.

\bibitem{DMV01}
W.~van Dam, M.~Mosca, and U.~Vazirani, \emph{How powerful is adiabatic quantum
  computation?}, Proc. 42nd IEEE Symposium on Foundations of Computer Science,
  pp.~279--287, 2001, \eprint{quant-ph/0206003}.

\bibitem{DS03}
W.~van Dam and G.~Seroussi, \emph{Quantum algorithms for estimating {G}auss
  sums and calculating discrete logarithms}, manuscript, 2003.

\bibitem{Dam88}
I.~B. Damg{\r{a}}rd, \emph{On the randomness of {L}egendre and {J}acobi
  sequences}, Advances in Cryptology - CRYPTO '88, Lecture Notes in Computer
  Science, vol. 403, pp.~163--172, 1990.

\bibitem{FGG07}
E.~Farhi, J.~Goldstone, and S.~Gutmann, \emph{A quantum algorithm for the
  {H}amiltonian {NAND} tree}, Theory of Computing \textbf{4} (2008), no.~1,
  169--190, \eprint{quant-ph/0702144}.

\bibitem{FG96}
E.~Farhi and S.~Gutmann, \emph{Analog analogue of a digital quantum
  computation}, Phys. Rev. A \textbf{57} (1998), no.~4, 2403--2406,
  \eprint{quant-ph/9612026}.

\bibitem{FG98}
\bysame, \emph{Quantum computation and decision trees}, Phys. Rev. A
  \textbf{58} (1998), no.~2, 915--928, \eprint{quant-ph/9706062}.

\bibitem{God05}
C.~Godsil, \emph{Association schemes}, lecture notes,
  \url{http://quoll.uwaterloo.ca/pstuff/assoc.pdf}.

\bibitem{GR02}
L.~Grover and T.~Rudolph, \emph{Creating superpositions that correspond to
  efficiently integrable probability distributions}, \eprint{quant-ph/0208112}.

\bibitem{Gro97}
L.~K. Grover, \emph{Quantum mechanics helps in searching for a needle in a
  haystack}, Phys. Rev. Lett. \textbf{79} (1997), no.~2, 325--328,
  \eprint{quant-ph/9706033},
  preliminary version in STOC 1996.

\bibitem{Ked06}
K.~S. Kedlaya, \emph{Quantum computation of zeta functions of curves}, Comput.
  Complex. \textbf{15} (2006), no.~1, 1--19, \eprint{math.NT/0411623}.

\bibitem{Lin92}
N.~Linial, \emph{Locality in distributed graph algorithms}, SIAM J. Comput.
  \textbf{21} (1992), no.~1, 193--201.

\bibitem{Llo96}
S.~Lloyd, \emph{Universal quantum simulators}, Science \textbf{273} (1996),
  no.~5278, 1073--1078.

\bibitem{LP96}
A.~Luis and J.~Pe{\v{r}}ina, \emph{Optimum phase-shift estimation and the
  quantum description of the phase difference}, Phys. Rev. A \textbf{54} 
  (1996), no.~5, 4564--4570.

\bibitem{MN07}
F.~Magniez and A.~Nayak, \emph{Quantum complexity of testing group
  commutativity}, Algorithmica \textbf{48} (2007), no.~3, 221--232,
  \eprint{quant-ph/0506265},
  preliminary version in ICALP 2005.

\bibitem{MNRS06}
F.~Magniez, A.~Nayak, J.~Roland, and M.~Santha, \emph{Search via quantum walk},
  Proc. 39th ACM Symposium on Theory of Computing, pp.~575--584, 2007,
  \eprint{quant-ph/0608026}.

\bibitem{MSS05}
F.~Magniez, M.~Santha, and M.~Szegedy, \emph{Quantum algorithms for the
  triangle problem}, Proc. 16th ACM-SIAM Symposium on Discrete Algorithms,
  pp.~1109--1117, 2005, \eprint{quant-ph/0310134}.

\bibitem{Mey96a}
D.~A. Meyer, \emph{From quantum cellular automata to quantum lattice gasses},
  J. Stat. Phys. \textbf{85} (1996), nos.~5/6, 551--574,
  \eprint{quant-ph/9604003}.

\bibitem{Mey96b}
\bysame, \emph{On the absence of homogeneous scalar unitary cellular automata},
  Phys. Lett. A \textbf{223} (1996), 337--340, \eprint{quant-ph/9604011}.

\bibitem{Moc06}
C.~Mochon, \emph{Hamiltonian oracles}, Phys. Rev. A \textbf{75} (2007), no.~4,
  042313, \eprint{quant-ph/0602032}.

\bibitem{MR01}
C.~Moore and A.~Russell, \emph{Quantum walks on the hypercube}, Proc. 6th
  International Workshop on Randomization and Approximation Techniques in
  Computer Science, Lecture Notes in Computer Science, vol. 2483, 
  pp.~164--178, 2002, \eprint{quant-ph/0104137}.

\bibitem{Reg06}
O.~Regev, \emph{Witness-preserving amplification of {QMA}}, lecture notes,
  \url{http://www.cs.tau.ac.il/~odedr/teaching/quantum_fall_2005/ln/qma.pdf}.

\bibitem{RS07}
B.~W. Reichardt and R.~{\v{S}}palek, \emph{Span-program-based quantum algorithm
  for evaluating formulas}, Proc. 40th ACM Symposium on Theory of Computing,
  pp.~103--112, 2008, \eprint{arXiv:0710.2630}.

\bibitem{RC02}
J.~Roland and N.~J. Cerf, \emph{Quantum search by local adiabatic evolution},
  Phys. Rev. A \textbf{65} (2002), no.~4, 042308, \eprint{quant-ph/0107015}.

\bibitem{RC03}
\bysame, \emph{Quantum-circuit model of {H}amiltonian search algorithms},
 Phys. Rev. A \textbf{68} (2003), no.~6, 062311,
  \eprint{quant-ph/0302138}.

\bibitem{Sch04}
W.~M. Schmidt, \emph{Equations Over Finite Fields: An Elementary Approach}, 2nd
  ed., Kendrick Press, 2004.

\bibitem{Sev03}
S.~Severini, \emph{On the digraph of a unitary matrix}, SIAM J. Matrix Anal.
  Appl. \textbf{25} (2003), no.~1, 295--300, \eprint{math.CO/0205187}.

\bibitem{SKW02}
N.~Shenvi, J.~Kempe, and K.~B. Whaley, \emph{A quantum random walk search
  algorithm}, Phys. Rev. A \textbf{67} (2003), no.~5, 052307, 
  \eprint{quant-ph/0210064}.

\bibitem{Sho97}
P.~W. Shor, \emph{Polynomial-time algorithms for prime factorization and
  discrete logarithms on a quantum computer}, SIAM J. Comput. \textbf{26}
  (1997), no.~5, 1484--1509, \eprint{quant-ph/9508027}.

\bibitem{Str06}
F.~W. Strauch, \emph{Connecting the discrete- and continuous-time quantum
  walks}, Phys. Rev. A \textbf{74} (2006), no.~3, 030301,
  \eprint{quant-ph/0606050}.

\bibitem{Sze04c}
M.~Szegedy, \emph{Quantum speed-up of {M}arkov chain based algorithms}, Proc.
  45th IEEE Symposium on Foundations of Computer Science, pp.~32--41, 2004,
  \eprint{quant-ph/0401053}.

\bibitem{Tul08}
A.~Tulsi, \emph{Faster quantum-walk algorithm for the two-dimensional spatial
  search}, Phys. Rev. A \textbf{78} (2008), no.~1, 012310,
  \eprint{arXiv:0801.0497}.

\bibitem{Wat01a}
J.~Watrous, \emph{Quantum simulations of classical random walks and undirected
  graph connectivity}, J. Comput. System Sci. \textbf{62} (2001), no.~2,
  376--391,
  \eprint{cs.CC/9812012}.

\bibitem{Wei48}
A.~Weil, \emph{On some exponential sums}, Proc. Natl. Acad. Sci. \textbf{34}
  (1948), no.~5, 204--207.

\end{thebibliography}
\end{document}